\documentclass[review]{elsarticle}

\usepackage{lineno,hyperref}
\modulolinenumbers[5]

\usepackage{lmodern}
\usepackage{lipsum}
\usepackage{graphicx}
\usepackage{qtree}
\usepackage{url}
\usepackage{verbatim}
\usepackage{amstext}
\usepackage{amsmath}
\usepackage{amsthm}
\usepackage{amsfonts}
\usepackage{amssymb}
\usepackage[utf8]{inputenc}

\usepackage{multirow}
\usepackage{multicol}
\usepackage{subfigure}
\usepackage{algorithm}
\usepackage{algorithmic}
\usepackage{mathtools}
\usepackage{multicol}
\usepackage{float}
\usepackage{fancyhdr}
\usepackage{titlesec}
\usepackage[center, small, sc ]{caption}
\usepackage{breqn}
\usepackage{fixltx2e}
\usepackage{rotating}
\usepackage{dsfont}

\usepackage[margin=0.67in]{geometry}
\newtheorem{definition}{Definition}
\newtheorem{lemma}{Lemma}

\journal{Comptes Rendus Physique}

\begin{document}

\begin{frontmatter}

\title{Evolution of Universes in Causal Set Cosmology}

%\author{Fay Dowker}
%\author{Stav Zalel}
%\address{Blackett Laboratory, Imperial College, London SW7 2AZ, UK}

\author{Fay Dowker}

\author{Stav Zalel}

\address{Blackett Laboratory, Imperial College, Prince Consort Road, London SW7 2AZ, UK.}
%\address[mysecondaryaddress]{Perimeter Institute}
%\address[mythirdaddress]{IQC}

\begin{abstract}
The causal set approach to the problem of quantum gravity is based on the hypothesis that spacetime is fundamentally discrete. Spacetime discreteness opens the door to novel types of dynamical law for cosmology and the Classical Sequential Growth (CSG) models of Rideout and Sorkin form an interesting class of such laws. It has been shown that a  renormalisation of the dynamical parameters of a CSG model occurs whenever the universe undergoes a Big Crunch-Big Bang bounce. In this paper we propose a way to model the creation of a new universe after the singularity of a black hole. We show that renormalisation of dynamical parameters occurs in a CSG model after such a creation event. We speculate that this could realise aspects of Smolin's Cosmological Natural Selection proposal. 
\end{abstract}


\end{frontmatter}

%\linenumbers

{\itshape ``This world and yonder world are forever giving birth:  every cause a mother, its effect a child. When the effect is born, it too becomes a cause and gives birth to wondrous effects. These causes are generation on generation, but it needs a very well lighted eye to see the links in the chain.''} {Rumi} \cite{Rumi}
 \vskip .3cm
\section{Introduction}

The causal set approach to the problem of 
quantum gravity  \cite{Bombelli:1987aa,Sorkin:1990bh,Sorkin:1990bj}
proposes that, of all the major concepts in our current 
best theories, the spacetime causal order 
from General Relativity and the path integral from quantum theory 
will survive the coming revolution.  
The main \textit{new} hypothesis in causal set theory is fundamental spacetime discreteness at the Planck scale. 
%These are the
%conceptual bases of causal set theory: the spacetime causal order and atomicity 
%furnish the kinematics and the path integral provides the framework 
%for the quantum dynamics. 

\subsection{The marriage of causal order and discreteness}

Spacetime causal order is a basic concept in General Relativity. 
The physics of General Relativity cannot be understood without reference to causal order, witness for example the widespread use of Penrose diagrams and the definition of a black hole. There are further good reasons for postulating that causal order is more fundamental than the other attributes of a spacetime,
 topology, differentiable structure and metric. Classic results in global causal analysis 
show that the causal order of a strongly causal  spacetime
determines its local null geodesics, its chronological structure and its topology \cite{Penrose:1972}, its differentiable structure \cite{Hawking:1976fe}
and its metric up to a
conformal factor \cite{HawkingEllis}. 
Causal order is therefore a unifying concept, containing $\frac{9}{10}$ of the full spacetime geometry (in $4$ dimensions) and lacking only information about local physical scale.

The other partner in the causal set marriage is spacetime discreteness which is perhaps the simplest way that the widely expected, Planck scale breakdown of the differentiable manifold description of spacetime can be made manifest. In the causal set approach to quantum gravity, discreteness
of spacetime is fundamental: the histories in the path integral 
 for full quantum gravity are discrete and the scale of the discreteness
 is close to the Planck scale. No continuum limit is 
 taken in the full theory and physics on large scales is a continuum \textit{approximation} to the 
 underlying theory. A continuum limit would lack some of the physics of the full theory \cite{Rideout:2000fh}.

Marrying causal order and discreteness straightforwardly results in a \textit{discrete order}, 
a discrete manifold 
whose only structure is a partial order relation $\prec$ on its elements. 
In the causal set approach to quantum gravity, the inner basis of spacetime
is hypothesised to be a discrete order or causal set \cite{Myrheim:1978,tHooft:1979,Bombelli:1987aa}. For further reviews of Causal Set theory see \cite{intro, discrete}. 
%This discrete manifold,
%called a causal set (or causetfor short), is a set, $C$ 
%with an order $\preceq$, such that:

%\noindent (1) if $x \preceq y$ and $y \preceq z$ then $x \preceq z$, 
%$\forall x,y,z \in C$ (transitivity);

%\noindent (2) if $x \preceq y$ and $y\preceq x$  
%then $x = y$, $\forall x,y \in  C$ (non-circularity);

%\noindent (3) for any pair of fixed elements $x$ and $z$ of $C$, the set 
%$\{y | x \prec y \prec z \}$ of elements lying between $x$ and $z$ is finite. 

%Of these conditions, the first two say that $C$ is a partially ordered
%set or poset.  The third 
%condition expresses {\textit {local finiteness}} and implies that the
%causal set is discrete. 

%\begin{figure}[ht]
%\begin{center}
%\includegraphics[width=0.8\textwidth]{256withlabels}
%\caption{The Hasse diagram of a causal setgenerated by 
%sprinkling into 1+1 dimensional Minkowski spacetime. The elements are the black dots and the
%blue edges are the links or nearest neighbour relations.}
%\label{256withlabels}
%\end{center}
%\end{figure}

 The way a causal set, $(C, \prec)$ gives rise to an approximating continuum 
spacetime  is that the order relation $\prec$ underpins the causal order
of $(M,g)$ and the physical scale missing from the causal order is furnished by the atomicity: the \textit{number} of elements 
in any portion of the causal set manifests itself, on average, as 
spacetime volume of the corresponding region of the approximating continuum. Number plus Order equals Geometry, in R. Sorkin's slogan. 
  
\subsection{A successful prediction from causal set cosmology}\label{dynamicsof}
 
The potential of causal set theory to make progress on cosmological 
questions has been demonstrated by the successful heuristic prediction of the order of magnitude of a fluctuating cosmological ``constant'', $\Lambda$, by Sorkin \cite{Sorkin:1990bj,Sorkin:1997gi}. To date this is the only successful prediction-in-advance from any quantum gravity theory. 
The argument starts by assuming that some 
as yet unknown dynamics drives $\Lambda$ towards
zero and any observed nonzero value of $\Lambda$ is a fluctuation about that mean value.  Fluctuations in 
$\Lambda$ are subject to an uncertainty relation with its canonically conjugate variable, the 
spacetime volume $V$. In the path integral over causal sets one does 
not sum over the cardinality of the causal sets, just as one doesn't sum over time in the path integral in quantum mechanics. The number of elements $N$ is held fixed and due to the statistical nature of the correspondence between number and volume, $V$ will be uncertain in value by $\sqrt{N} \sim \sqrt{V}$ where $V$ is measured in fundamental Planckian units. So
$\Lambda \sim \Delta \Lambda \sim (\Delta V)^{-1}
\sim \pm \frac{1}{\sqrt{V}}\sim\pm 10^{-120} $ in natural units,  where we 
have used the volume of the observable universe for $V$. 

Computer simulations of a phenomenological model of the fluctuations of $\Lambda$ confirmed
the earlier prediction \cite{Ahmed:2002mj}. Certain ad hoc aspects of this
phenomenological model have been alleviated in \cite{Ahmed:2012ci} 
but one serious one remains: $\Lambda$ is forced to be homogeneous.
It has been argued that if $\Lambda$ were allowed to fluctuate 
in space this would conflict with $\frac{\delta T}{T} \sim 10^{-5}$ constraints from
the CMB \cite{Barrow:2006vy,Zuntz:2008zza}. This conclusion, however, assumes that General Relativity holds and the fluctuations couple locally to the metric. The original homogeneous fluctuating $\Lambda$ model is highly nonlocal 
in time, however, and it is possible that a nonlocal model of fluctuating inhomogeneities in $\Lambda$ can be compatible with CMB constraints. If this turns out to be the 
case then the fluctuations in $\Lambda$ might themselves be able to act 
as the seeds for structure formation in the early universe.

There is currently a window of opportunity for causal set $\Lambda$ with increasing tension between high and low redshift measurements of the Hubble parameter and the standard constant-$\Lambda$CDM cosmological model. In particular, two measurements of the Hubble parameter at redshift of around $z = 2.34$ by the Baryon Oscillation Spectroscopic Survey (BOSS) which uses the Baryon Acoustic Peak as a standard ruler, agree with each other and give a best fit value of $\rho_\lambda(z = 2.34)$ which is negative
 \cite{Font-Ribera:2013wce,Delubac:2014aqe}. The significance of these two
 results is not too high by themselves but nevertheless this is exciting 
 because the causal set model of fluctuating $\Lambda$ implies that 
 it will have been negative in the past. Further observations of model independent values of $H(z)$ will be of great importance in deciding if the two BOSS results are the heralds of the end of a constant Cosmological Constant.   

\subsection{Notation and terminology used in this work}

\par 
We use the term \textbf{causet} as shorthand for causal set. 
Consider a causet $(C,\prec)$ with elements $x$ , $y$. If $x \prec y$ we say ``$x$ precedes $y$'' or ``$x$ is below $y$'' or ``$y$ is above $x$'' or ``$x$ is an ancestor of $y$'' or  ``$y$ is a descendant of $x$.'' If $x\ne y$ are unrelated by $\prec$ we write $x\natural y$ and say ``$x$ and $y$ are unrelated''. The order is irreflexive: $x\nprec x$. 

\par The \textbf{past} of a subset $A$ of $C$ 
is the subset $P(A):=\left\{y\in C\,|\,
\exists x\in A\ s.t.\  y\prec x\right\}$. The \textbf{future} of $A$ is the subset $F(A):=\left\{y\in C\,|\, \exists x\in A\ s.t.\  y\succ x\right\}$. 
%The definitions $P(A)=P(A) \setminus A$ and $F(A)=F(A) \setminus A$ will be useful.
 For a single element, with a slight abuse of notation, we write $P(x)$ ($F(x)$) as the past (future) of $\{x\}$.
 % and similarly for $P(x)$ and $F(x)$.

\par An element $x\in C$ is a \textbf{maximal (minimal) element} of $C$ if $\not\exists$ $y\in C$ such that $y\succ x$ ($y\prec x$). 

%\par A \textbf{chain} is a subset of $C$ in which any two elements are related. A \textbf{maximal chain} in $C$ is a chain such that there is no element $x \in C$ not in the chain that is related to all elements in the chain. An \textbf{antichain} is a subset of $C$ in which all elements are unrelated to each other. A \textbf{maximal antichain} in $C$ is an antichain such that there is no element $x \in C$ not in the antichain that is not related to all elements in the antichain.
\par A \textbf{link} is a relation $x \prec y$ such that $\not\exists$ $z$ such that $x \prec z \prec y$.  If $x \prec y$ is a link then $x$ is an \textbf{immediate ancestor} of $y$ and $y$ is a \textbf{direct descendant} of $x$. 

%\par A \textbf{path} is an increasing sequence of elements, each related to the next by a link.

A \textbf{future set} is a subset which contains its own future. A \textbf{past set} is a subset which contains its own past.
A \textbf{partial stem} is a past set of finite cardinality.
% A \textbf{full stem} is a partial stem such that every element of its complement lies to the future of at least one of its maximal elements. 

\par A \textbf{post} is an element which is causally related to every other element in the set.  
\par A finite causet can be represented as a \textbf{Hasse diagram} in which each element is a dot or node and if $x\prec y$ is a link then an upward-going 
edge is drawn from $x$ to $y$. 

\section{New universes in causal set theory} \label{physkin}

Smolin's idea of Cosmological Natural Selection \cite{universeevolve} uses the 
conjecture that a new universe\footnote{This use of the word `universe' is a misnomer since we must contemplate the existence of multiple
universes. Nevertheless, it is widespread and so we will use it.} can come into being at the singularity of a black hole. This natural and appealing idea has also been proposed by S. Hawking and others as
``baby universes.'' Smolin also proposes that the coupling constants of physical theory are different in the child universe than in the parent. 
 These conjectures will eventually have to be proved within a full theory of quantum gravity but 
for now, we can ask: how could such an event in principle be described within a particular approach? In causal set theory, we propose the following. 

\begin{figure}[h]
  \centering
	\includegraphics[width=0.8\textwidth]{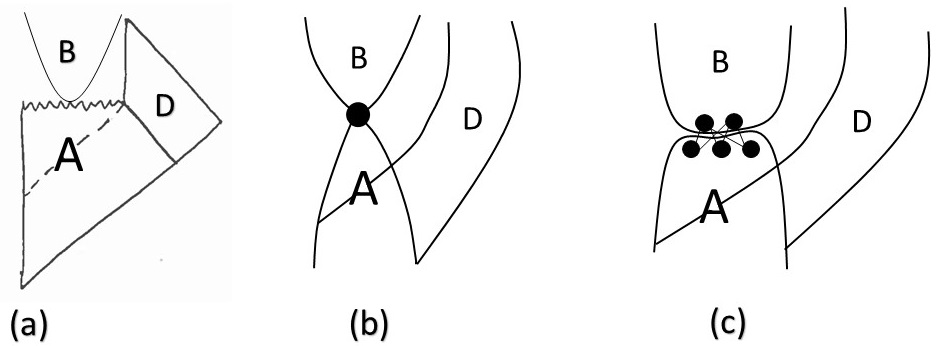}
	\caption{A black hole gives rise to a new universe. The new universe is $B$, the past of the singularity is $A$, and
	$D$ is the region unrelated to $B$  (a) A ``continuum'' Penrose diagram (b) A sketch of a causet in which the singularity is a single element (c) A sketch of a causet  in which the singularity is a partition in which all minimal elements of region $B$ are above all maximal elements of region $A$}
	\label{bh_diagrams}
\end{figure}

Consider the Penrose diagram for the formation and evaporation of a black 
hole in figure \ref{bh_diagrams}(a) with the  new
universe $B$ sketched heuristically above the singularity. We propose two 
ways to model this as a causet. Figure \ref{bh_diagrams}(b) represents a causet in which the singularity\footnote{The word ``singularity''  is being used in the sense that it is an event where General Relativity 
breaks down and that cannot be described in terms of a classical Lorentzian geometry. But, the idea is, it \textit{can} be described  in the underlying discrete 
theory as part of a causal set and so it is not, strictly, ``singular'' in the full theory at all.} is modelled as a single element of the set called a \textit{partial post} \cite{glossary}. In the causet sketched in figure \ref{bh_diagrams}(c), the singularity is not an element of the causet but 
a partition of the causet we call a \textit{partial break}. 

Let us formalise these concepts. 
\begin{definition}
A partial post is an element $y$ of causet $C$ such that,  if $y\prec x$ and $y\natural z$, then $x\natural z$.
\end{definition}
This is illustrated in figure \ref{bh_diagrams}(b) in which $A = P(y)$, $B = F(y)$ and $D = C \setminus (A\cup B\cup\{y\})$. 

\begin{definition}
A partial break in a causet $C$ is an ordered pair of nonempty subsets $(A,B)$ of C that satisfies
\begin{itemize}
\item  $a\in A$, $b\in B$ $\implies a\prec b$;
\item  $B$ is a future set;
\item $A \cup B$ is a past set.
\end{itemize}
\end{definition}
We call $A$ the past of the partial break and $B$ the future of the partial break. 
This is illustrated in figure \ref{bh_diagrams}(c) where 
$D = C \setminus (A\cup B)$. 

Special cases of these concepts can be defined when the region $D$ shown in figures  \ref{bh_diagrams}(b) and \ref{bh_diagrams}(c) is empty:
\begin{definition}
A post is a partial post, $y$, such that the whole causet $C = P(y)\cup F(y)\cup \{y\}$.
\end{definition}
\noindent Note, this is an equivalent to the definition of a post, given previously, as an element that 
is related to all other elements in $C$. A post models the collapse to a single 
element, and subsequent re-expansion of the whole 
spacetime. 
\begin{definition}
A break is a partial break $(A,B)$ such that the whole causet $C = A\cup B$.
\end{definition}
\noindent A break is not collapse to a single point but it 
is something that does not have a continuum interpretation: a causet with a break could not be 
embedded in a Lorentzian spacetime such that the causet order and spacetime order of the 
embedded elements agree. A break is therefore eligible to be considered
as a ``singularity'' from the continuum point of view. 

Now, there is a partial break above and below a partial post:
\begin{lemma}\label{postimpliesbreaks}
If $y$ is a partial post in $C$, then the following pairs of subsets are both partial breaks: $(P(y), F(y)\cup\{y\})$ and $(P(y)\cup\{y\}, F(y))$.
\end{lemma}
\begin{proof} Let $y$ be a partial post. Consider the pair $(P(y), F(y)\cup\{y\})$. The first two conditions  for a partial break are clearly satisfied. Consider $z$, the ancestor of an element of $P(y) \cup F(y)\cup\{y\}$. If $z$ is an ancestor of $y$ itself then $z\in P(y)$. If $z$ is an ancestor of an element of $P(y)$ then it is an ancestor of $y$ and in $P(y)$. If $z$ is an ancestor of an element of $F(y)$ then $z$ cannot be 
unrelated to $y$. So $z \in P(y) \cup F(y)\cup\{y\}$. So $P(y) \cup F(y)\cup\{y\}$ is a past set. 

Similarly for $(P(y)\cup\{y\}, F(y))$.
\end{proof}

Our intuition points towards a partial post as being the appropriate representation of 
a singularity caused by gravitational collapse: a single spacetime atom corresponding to one Planck 
unit of spacetime volume, what could be more ``collapsed'' than that? However, one could also
make a case for a break to be even more singular: a break is, in some sense,
the nothingness between the end of a portion of an old universe 
and the beginning of the new. Of course, we do not know enough about quantum gravity to make a
final judgment, but in any case, we don't need to decide between them for the purpose of the current paper.
Since, by Lemma \ref{postimpliesbreaks},  the existence of a partial post is a stronger condition than 
the existence of a partial break and we can focus on the partial break as our
model of the birth of a new universe. Whatever follows from the existence of 
a partial break will also follow from the existence of a partial post. 

So the question then is, if a partial break occurs in a causet 
that describes our spacetime, what consequences does that have? To 
answer this question we need to embed it in a dynamical model for 
causets. 

\section{Classical Sequential Growth}

The Classical Sequential Growth models of Rideout and Sorkin \cite{CSG} are  stochastic processes in which a past-finite causet 
grows by the continual birth of new elements which respect the discrete 
analogue of general covariance and a causality condition ensuring that the growth of one part of the causet does not depend on the structure 
of the causet spacelike to it. 

For a given CSG model, the growth process consists of countably many stages and the causet is finite at each stage. At each stage a single new element is added. Suppose, causet $C_n$, with cardinality $n$, has already grown.  The transition from this parent causet, $C_n$, to child causet, $C_{n+1}$, with one more element is called \textbf{stage $n$}. The new element chooses, with a certain probability, a partial stem in $C_n$ as its past; this partial stem is the \textbf{precursor} of the transition. A \textbf{spectator} of the transition is an element which is not an element of the precursor; it is unrelated to the new element.  
\begin{figure}[h]
  \centering
	\includegraphics[width=0.5\textwidth]{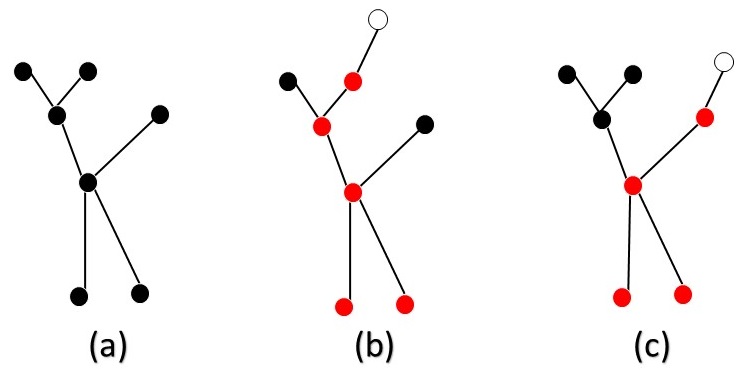}
	\caption{(a) The parent causet, $P$, of cardinality 7; (b) possible child $Q$; (c) possible child $R$. In each of 
	(b) and (c) the new element is the empty circle, the precursor set of each transition is red and the spectator elements of the transition are black.}
	\label{CCC}
\end{figure}

An example of stage $7$ is shown in figure \ref{CCC}. There are two possible transitions shown. In each, (a) is the parent causet, and (b) and (c) are two possible children.

A CSG model is specified by a countable set of non-negative constants,
\begin{equation*}
 t_0, t_1, t_2, t_3, t_4\dots,
 \end{equation*}
where $t_k$ is the relative probability that the newly born element chooses a particular set of ancestors which has cardinality $k$. The probability of the transition from $C_n$ to $C_{n+1}$ is \cite{CSG}
\begin{equation}
\label{probability_of_transition}
Prob(C_n \to C_{n+1}) = \frac{\lambda \left(\varpi,m\right)}{\lambda \left(n,0\right)}
\end{equation}
where
\begin{equation}
\lambda \left(\varpi,m\right)=\sum_{k=m}^{\varpi} \binom{\varpi-m}{k-m} t_k
\label{relative_prob_of_trans_PPXXYY}
\end{equation}
and $\varpi$ is the cardinality of the precursor and $m$ is the number of maximal elements of the precursor. 
The probability of growing, by stage $n$, the particular causal set $C_n$ is given by the product of the probabilities for the $n$ transitions that produce $C_n$. 
From the form of the transition probability we see that the space of CSG models is a projective space: the set of parameters $(t_k)$ 
and the set $(\alpha t_k)$, where $\alpha>0$ is a constant, define the same CSG model. 

Note that the denominator of (\ref{probability_of_transition}) is
the same for all the possible transitions at stage $n$. The numerator, $\lambda \left(\varpi,m\right)$, is therefore the \textit{relative probability} of the transition $C_n \to C_{n+1}$ -- relative to the other transitions that are possible from $C_n$. For the proof that the sum of the relative probabilities over all possible transitions from $C_n$ equals the denominator, $\lambda \left(n,0\right)$, see \cite{CSG}. 

This description of a CSG model treats the births as if they happen in a total order and each element of the resulting causet is labelled by the stage at which it is born. Part of this information -- the order of birth of any two elements that are spacelike to each other in the resulting causal set --  is pure gauge and unphysical. Indeed, in any particular CSG model, two order-isomorphic causal sets with different labellings have equal probability of growing:  a condition called ``discrete general covariance.'' Each CSG model also satisfies a causality condition called ``Bell causality'' which means the growth of the causet is not influenced by its structure spacelike to the part of the causet that is growing.  For example, consider a parent causet $P$ and two possible children $P\rightarrow Q$ and $P\rightarrow R$ as shown in figure \ref{CCC}. The element at the top left corner of $P$ is a spectator of \textit{both} transitions. We can delete that spectator from $P$, $Q$ and $R$ to form causets $\bar{P}$, $\bar{Q}$ and $\bar{R}$ as shown in figure \ref{BBB}. 
The Bell Causality condition is:
\begin{equation}
\frac{Prob(P\to Q)}{Prob(P\to R)}=\frac{Prob(\bar{P}\to \bar{Q})}{Prob(\bar{P}\to \bar{R})}\,.
\end{equation}
The presence, or absence, of any spectator has no effect on the relative probabilities of transitions. This will be important later.

\begin{figure}[h]
  \centering
	\includegraphics[width=0.5\textwidth]{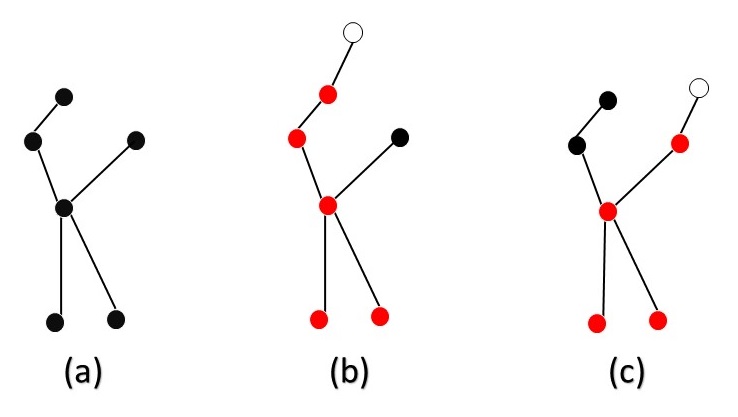}
	\caption{The transition in figure \ref{CCC} without the spectator element (a) Parent causet $\bar{P}$; (b) child causet $\bar{Q}$; (c) child causet ${\bar{R}}$}
	\label{BBB}
\end{figure}

\subsection{Cosmic Renormalisation}
Sorkin showed that, in a CSG model $(t_k)$, if the causal set that grows contains a post -- a Big Crunch-Big Bang event -- the effective dynamics of the causal set after the post is governed by a different CSG model with a renormalised set of constants $\{t'_k\}$ \cite{Sorkin:1998hi}. 

The crucial point is actually that there is a \textit{break} in the causal set and we will redo the proof here in the more general case of a break but the essential idea is the same.  

\begin{lemma} In a CSG model, the growth of the future of a break is governed by a CSG dynamics with renormalised parameters. 
\end{lemma} 
\begin{proof} Suppose that at stage $N$, a causal set $A$ with cardinality $N$ and with $r$ maximal elements has grown. Let there be a break between $A$ and the rest of the causal set.  Following Sorkin we refer to the causal set after the break as the \textit{current era} 
and consider it as a new universe. The transition probabilities for the growth of the current era are given by the transition probabilities in the original CSG model \textit{conditioned on} the existence of a break after $A$. 

The condition that there is a break after $A$ means that all elements born in the current era must have a precursor set that includes the whole of $A$. This restricts the allowed transitions. The most convenient quantities to work with are therefore the relative probabilities of the allowed transitions  and it is understood that the normalisation factor will always be the sum of the relative probabilities over the allowed transitions at each stage. 

Consider element $y$ born in the current era at stage $N+n_c$ in the full dynamics, \textit{i.e.} at effective stage $n_c$ in the current era. Let
$y$ have $\varpi$ ancestors  and $m$ immediate ancestors \textit{in the current era}.  

If $y$ is a minimal element of the current era then $\varpi=m=0$, and $y$ has $N$ ancestors and $r$ immediate ancestors in the whole causal set including $A$. The relative probability of the transition in which $y$ is born equals
\begin{equation}
\label{originary}
\lambda \left(N,r\right)\,.
\end{equation}

If $y$ is a non-minimal element of the current era then $n_c \ge 1$ and $y$ has $\varpi +N$ ancestors and $m$ immediate ancestors, with $\varpi\not= 0$, $m\not= 0$.
 The relative probability of the transition in which $y$ is born is equal to $\lambda(\varpi+N,m)$.
 
 We have 
 \begin{equation}
\begin{split}
\label{proof_of_relation_of_t_and_t_tilde}
\lambda \left(\varpi+N,m\right)
&=\sum_{s=m}^{\varpi+N} \binom{\varpi+N-m}{s-m} t_s\\
&=\sum_{s=m}^{\varpi+N} \sum_{k+l=s} \binom{\varpi-m}{k-m} \binom{N}{l} t_s\\
&=\sum_{k+l=m}^{\varpi+N} \binom{\varpi-m}{k-m} \binom{N}{l} t_{k+l}\\
&=\sum_{k=m}^{\varpi} \sum_{l=0}^{N} \binom{\varpi-m}{k-m} \binom{N}{l} t_{k+l}\\
&=\sum_{k=m}^{\varpi} \binom{\varpi-m}{k-m} \sum_{l=0}^{N} \binom{N}{l} t_{k+l}\\
&=\sum_{k=m}^{\varpi} \binom{\varpi-m}{k-m} \tilde{t}_k
\end{split}
\end{equation}
where the identity 
 \begin{equation} \sum_{i+j=k} \binom{m}{i}\binom{N}{j}=\binom{m+N}{k} 
 \end{equation}
was used and we have defined
\begin{equation}
\tilde{t}_k:= \sum_{l=0}^{N} \binom{N}{l} t_{k+l} \hspace{3mm} \forall \hspace{3mm} k\ge 1\,.
\end{equation}
Defining $\tilde{t}_0 : = \lambda \left(N,r\right)$, and 
$\tilde{\lambda}(\varpi,m):=\lambda(\varpi+N,m)$,
we see that the relative probabilities of the allowed transitions in the current era have the form of a CSG model with new, ``renormalised'' parameters $(\tilde{t}_k)$. 
\end{proof}

\section{Branching universes}

Our purpose is to model the branching off of a child universe as a partial break. Let us consider, then, a CSG model $(t_k)$  in which a causet, $C$, grows which has a partial break. We will see that we have already done all the work needed to determine the effective dynamics of the child universe. 

If a causet grows in which there is a partial break, then, at some stage of the process the past of the partial break -- a set $A$ of cardinality $N$, say --  will have grown. 
The growth of the child universe, \textit{i.e.} the future of the partial break, is conditioned on the existence of the break of which $A$ is the past. 
Let $D$ be the set of elements that are spacelike to the child universe
as illustrated in figure \ref{bh_diagrams}. 
Consider a transition in which an element of the child universe is born. Any already existing element of $D$ will be a spectator to this transition. 
This means that, by the Bell causality condition, 
 the relative probabilities of the transitions which grow the new universe are 
the same as those with the elements of $D$ removed. 

 In other words the growth of the child universe to the future of the partial break is governed by the \textit{same} stochastic law as for the growth of the future of a break with past $A$. So, 
the growth of the child universe is governed by a CSG model in which the parameters, $(\tilde{t}_k)$ are given as before by 
\begin{align}
\tilde{t}_k&:= \sum_{l=0}^{N} \binom{N}{l} t_{k+l} \hspace{3mm} \forall \hspace{3mm} k\ge 1\\
\tilde{t}_0& : = \lambda \left(N,r\right)
\end{align}
where $N$ is the cardinality of $A$ and $r$ is the number of maximal elements of 
$A$. 

The question of whether some of the concepts of Smolin's cosmic natural selection scenario can be realised within a dynamical theory of causal sets is therefore transformed into concrete questions about the class of CSG models: are there partial breaks in causal sets grown in any CSG models and, if there are, what is the result of the renormalisation of the parameters given above? 
 
For the special case of posts, we already have some answers. 
For example, each member of the one parameter family of CSG models known as transitive percolation, for which $t_k = t^k$ where $t>0$ is a real number, almost surely produces a causet with infinitely many posts \cite{alon}. 
Brightwell claims there is a larger class of CSG models in which the same is true \cite{Brightwell:2009}. Transitive percolation models form a line of fixed points of the renormalisation transformation after a post. 
The ``RG flow'' induced by the 
existence of infinitely many posts and the question of whether and in what sense there is a basin of attraction around the line of transitive percolation was studied in \cite{Martin:2000js}. Sorkin has argued for the  potential relevance of this cosmic renormalisation to the ``large number'' puzzles of cosmology including the so-called ``flatness problem''. This is the observation that if our observed universe is evolved back in time using the Friedmann equation with standard 
 assumptions about the matter content of the universe, then when the Hubble parameter 
 $H= \frac{\dot{a}}{a}$ becomes of order one in natural units -- \textit{ i.e.} spacetime curvature becomes Planckian 
 -- the radius of curvature of 3 dimensional \textit{space}  is about 28 orders of magnitude bigger than one. 
Inflationary scenarios do not explain this because they themselves
 have to be fine tuned to produce this number. Sorkin's argument, in contrast, is that the universe could be \textit{self-tuning}  and that the underlying dynamics could be ``natural'' -- i.e. contain only numbers of order one. The suggestion is that the occurrence of a very large number of cycles of cosmic expansion and collapse, punctuated by posts, could result in an effective dynamics for our current era such that the cosmic expansion after the latest post would have resulted in a  large, almost flat, 3 dimensional space at the end of the Planck era  \cite{Sorkin:1998hi}. More work is needed to see if this proposal can be realised. 

For Smolin's Cosmic Natural Selection, a model of child universes created at the singularities of 
black holes requires partial breaks or partial posts that are \textit{not} posts. 
It is an open question if these can occur in any CSG model and, if it turns out they cannot, exploring the notion of black holes being the birthplace of new universes will need a different kind of causal set dynamics, 
perhaps a fully quantal dynamics. 

\section*{Acknowledgment}
This work is supported by STFC grant ST/L00044X/1. SZ is grateful for support from the Alan Howard Foundation and the Kenneth Lindsay Scholarship Trust. 

%\section*{References}

\bibliography{../bibliography/refs}

\end{document}